\documentclass[12pt]{amsart}
\usepackage{a4wide}
\usepackage[dvips]{epsfig}
\usepackage{amscd}
\usepackage{amssymb}
\usepackage{amsthm}
\usepackage{amsmath}
\usepackage{latexsym}

\theoremstyle{plain}
\newtheorem{thm}{Theorem}[section]
\theoremstyle{plain}

\newtheorem{prop}[thm]{Proposition}

\theoremstyle{definition}

{%
\setcounter{enumi}{0}

\begin{enumerate}}%
{\end{enumerate} }
{%
\setcounter{enumi}{0}

\begin{enumerate}}%
{\end{enumerate} }

\numberwithin{equation}{section}
 \allowdisplaybreaks

\title[Hedging Quanto Options]
 {On the sensitivity analysis of energy quanto options.}

\date{\today}

\begin{document}
\author{ Rodwell Kufakunesu }

\address{Department of Mathematics and Applied Mathematics
, University of Pretoria, 0002, South Africa}

\email{rodwell.kufakunesu@up.ac.za}

\author{ Farai Julius Mhlanga }

\address{Department of Mathematics and Applied Mathematics,
University of Limpopo, Private bag X1106, Sovenga, 0727, South Africa}

\email{farai.mlanga@ul.ac.za}






\begin{abstract}
In recent years there has been an advent of quanto options in energy markets.
The structure of the payoff is rather a different type from other markets since it is written as a product of an underlying energy index and a measure of temperature.
In the HJM framework, by adopting the futures energy dynamics, we use the Malliavin calculus to derive the delta and the cross-gamma expectation formulas. This work can be viewed as an extension of the work done, for example by Benth \emph{et al.} \cite{BenthLange}.
\end{abstract}

\maketitle
\section{Introduction}
The paper investigates hedging of the energy quanto options using the Malliavin
Calculus approach by Nualart \cite{Nualart}. This method has shown that it outperforms the finite difference approach when it comes to discontinuous payoffs, see Benth \emph{et al.} \cite{BenthGroth}. Quanto options in the equity market differ
from those designed for energy markets by the structure of their payoffs. The
energy quanto option has a product payoff which is structured in such a way that
it takes advantage of the high correlation between energy consumption and
certain weather conditions thereby enabling price and weather risk to be
controlled simultaneously, refer Caporin \emph{et al.} \cite{Pres}. On the other hand, the equity quanto has
a normal structure. Ho \emph{et al.} \cite{Ho} noted that quanto options, in general,
are better hedgers than a simple combination of plain vanilla options. In energy
markets, they give exposure to the volumetric risk input of weather conditions
on energy prices, see Zhang \cite{Zhang}.

Heath \emph{et al.} \cite{Heath} introduced the so-called Hearth-Jarrow-Merton (HJM) approach in fixed income markets where the dynamics of the forward rates are directly specified, see Benth \emph{et al.} \cite{BenthElectricity}. The fact that most contracts in energy markets are settled in futures and forward, the framework was later on in 2000 adopted in this market by Clewlow and Strickland \cite{Clewlow}.
There have been few papers in literature analysing the hedging of this quanto option product.
Benth \emph{et al.} \cite{BenthLange} recently studied the pricing and hedging of
quanto energy options in this framework basing on both the
spot and the futures products as the underlying processes. The authors derived
analytic expressions for the energy delta, the temperature delta, and the
cross-gamma hedging. If the payoff functions are discontinuous then their
hedging method fails. In this paper, we derive the so-called derivative free
hedging formulae using a much more powerful tool: the Malliavin calculus. The
Malliavin calculus technique has been used by many authors in hedging equity
derivative products, see for example,
(Benth \emph{et al.} \cite{BenthElectricity}, Benth \emph{et al.} \cite{BenthDahl}, Di Nunno \emph{et al.} \cite{DNunnoOksendal}, Fourni\'{e} \emph{et al.} \cite{Fournie1,Fournie2}, Karatzas \emph{et al.} \cite{KaratsasOkone}, Mhlanga \cite{Mhlanga}). In all
these references the methods were not applied in a product payoff structure such
as ours and with an interval delivering period. Our results can be viewed as an
extension of Benth \emph{et al.} \cite{BenthLange}.

We only focus on deriving the delta and the cross-gamma hedging expectation formulae.
The paper is organised as follows. Section 2, reviews the structure of the
quanto option as in Benth \emph{et al.} \cite{BenthLange}. We present the futures asset dynamics general diffusion models under the Heath-Jarrow-Morton (HJM) framework.  In Section 3, we review the
necessary tools from Malliavin calculus to be applied in our proofs. In Section 4, the `Greeks' that is, the delta and the cross-gamma are derived in the independent case, when the correlation value is zero and in Section 5 the correlation case, the Greek formulae are obtained. The residual risk is also discussed. Some examples are provided in Section 6. Finally, in Section 7, we conclude.
\section{The Contract Structure and Pricing of Quanto Options}
In this section, we review the commodity quanto pricing, see, for example, Benth \emph{\emph{et al.}} \cite{BenthLange}
and in particular, we follow their notations therein. The quanto has a payoff function $S$ given
by:
$$
S=(T_{var}-T_{fix})\times(E_{var}-E_{fix}),
$$
where $T_{var}$ represents some variable temperature measure, $T_{fix}$
represents some fixed temperature measure, and $E_{var},E_{fix}$ are the
variable and fixed energy price, respectively. To avoid the downside risk on
this quanto contract it has been reported in Benth \emph{et al.} \cite{BenthLange} that for hedging
purposes, it is reasonable to buy a contract with optionality. In the
temperature market of Chicago Mercantile Exchange (CME), for contracts are
written on the aggregated amount of heating-degree days (HDD) and cooling-degree
days (CDD). The temperature index is used as the underlying. The HDD (similarly
the CDD) over a measurement period $[\tau_1,\tau_2]$ is defined by:

\begin{equation}
I_{[\tau_1,\tau_2]}:=HDD(t)=\max(c-T(t),0),
\end{equation}
where $T(t)$ is the mean temperature on day $t$, and $c$ is the pre-specified
temperature threshold (eg., $65^0$F or $18^0$C). If the contract is specified as
the accumulated HDD over $[\tau_1,\tau_2]$ we have :
 \begin{equation}
I_{[\tau_1,\tau_2]}:=\sum_{t=\tau_1}^{\tau_2}HDD(t)=\sum_{t=\tau_1}^{\tau_2}
\max(c-T(t),0),
\end{equation}
analogously for CDD.\\
We note that quanto options have a payoff function that is a function of two underlying assets, temperature and price. We focus on quanto options with payoff function $f(E,I)$ where $E$ is an index of the energy price and $I$ is an index of temperature. The energy index $E$ over a period $[\tau_1,\tau_2]$ is given as
an average spot price by:
\begin{equation}\label{energy}
E=\frac1{\tau_2-\tau_1}\sum_{u=\tau_1}^{\tau_2}S_u,
\end{equation}
where $S_u$ is the energy spot price. In addition, we assume that the temperature index is defined by
\begin{equation}\label{tempera}
I=\sum_{u=\tau_1}^{\tau_2}g(T_u),
\end{equation}
where $T_u$ denotes the temperature at time $u$ and $g$ some function. For example, for a quanto option involving the HDD index, we choose $g(x)=\max(x-18,0)$. To price the quanto option exercised at
the time $\tau_2$, its arbitrage-free price at time $t\leq\tau_2$ becomes
\begin{equation}\label{OptionQ}
 C_t=e^{-r(\tau_2-t)}\mathbb E_t^{\mathbb
Q}\left[f\left(\frac1{\tau_2-\tau_1}\sum_{u=\tau_1}^{\tau_2}S_u,~\sum_{u=\tau_1}
^{\tau_2}g(T_u)\right)\right],
\end{equation}
where $r>0$ represents a constant risk-free interest rate and $\mathbb
E_t^{\mathbb Q}$ is the expectation operator with respect to $\mathbb Q$,
conditioned on the market information at time $t$ given by the filtration
$\mathcal F_t$. Following \cite{BenthLange}'s argument on the relationship
between the quanto option and the futures contract on the energy ad temperature indexes $E$ and $I$, we note that the price at time $t\leq \tau_2$ of a futures contract written on some energy price with delivery period $[\tau_1,\tau_2]$ is given by
\begin{equation}\label{energyfuture}
 F^E(t;\tau_1,\tau_2)=\mathbb E_t^{\mathbb
Q}\left[\frac1{\tau_2-\tau_1}\sum_{u=\tau_1}^{\tau_2}S_u\right]\,.
\end{equation}
At $t=\tau_2$ we have:
\begin{equation}
F^E(\tau_2;\tau_1,\tau_2)=\frac1{\tau_2-\tau_1}\sum_{u=\tau_1}^{\tau_2}S_u\,.
\end{equation}
This means that the future prices is exactly equal to what is being delivered. Applying the same argument to the temperature index, with price dynamics denoted by $F^I(t,\tau_1,\tau_2)$, the quanto option price $C_t$ can be written as:
\begin{eqnarray}
\label{quantooptionprice}
C_t&=&e^{-r(\tau_2-t)}\mathbb E_t^{\mathbb
Q}\left[f\left(\frac1{\tau_2-\tau_1}\sum_{u=\tau_1}^{\tau_2}S_u,~\sum_{u=\tau_1}
^{\tau_2}g(T_u)\right)\right]\\\nonumber
&=&e^{-r(\tau_2-t)}\mathbb E_t^{\mathbb
Q}\left[f\left(F^E(\tau_2;\tau_1,\tau_2),F^I(\tau_2;\tau_1,\tau_2)\right)\right]
\,.
\end{eqnarray}
The advantage of writing the quanto option price as in
Eq.\eqref{quantooptionprice} is that futures are traded financial assets. Let $\bar{K}_E,\bar{K}_I$ denote the high strikes for the energy  and temperature indexes, respectively and $\underline{K}_E,\underline{K}_I$ denote the low strikes for the energy  and temperature indexes, respectively. Now we can define the payoff function
$$
p(F^E(\tau_2;\tau_1,\tau_2),F^I(\tau_2;\tau_1,\tau_2),\bar{K}_E,\bar{K}_I,
\underline{K}_E,\underline{K}_I):=p$$ so that
\begin{eqnarray*}
p&=&\alpha\times[\max(F^E(\tau_2;\tau_1,\tau_2)-\bar{K}_E,0)
\times \max(F^I(\tau_2;\tau_1,\tau_2)-\bar{K}_I,0)\\
\quad&+&\max(\underline{K}_E-F^E(\tau_2;\tau_1,\tau_2),0)\times\max(\underline{K}
_I-F^I(\tau_2;\tau_1,\tau_2),0)]\,,
\end{eqnarray*}
where $\alpha$ is the contractual volume adjustment factor. As in Benth et al. \cite{BenthLange}, for illustration
purposes, we consider the product call structure with the volume adjuster $\alpha$ normalized to $1$, that is, we consider the price of an option with the following payoff function:
\begin{equation}
\label{payoffforward}
\hat{p}=\max(F^E(\tau_2;\tau_1,\tau_2)-\bar{K}_E,0)\times
\max(F^I(\tau_2;\tau_1,\tau_2)-\bar{K}_I,0)\,,
\end{equation}
and the quanto option at time $t$ is given by:
\begin{equation*}
 C_t=e^{-r(\tau_2-t)}\mathbb E_t^{\mathbb
Q}\left[\hat{p}(F^E(\tau_2;\tau_1,\tau_2),F^I(\tau_2;\tau_1,\tau_2),
\bar{K}_E,\bar{K}_I)\right]\,.
\end{equation*}

\subsection{The Asset Dynamics}

We use the HJM risk-neutral dynamics of the forward
contract at time $t$. Consider the general diffusion futures model under the
risk-neutral measure $\mathbb{Q}$ be given as :

\begin{eqnarray}
 dF^E(t;\tau_1,\tau_2) &=&\sigma_E(t, F^E(t;\tau_1,\tau_2))dW^E(t),\\
 dF^I(t;\tau_1,\tau_2) &=& \sigma_I(t,F^I(\tau;\tau_1,\tau_2))dW^I(t)\,.
\end{eqnarray}
with $F^E(0;\tau_1,\tau_2)>0$ and $F^I(0;\tau_1,\tau_2)>0$ where $\sigma_E$, $\sigma_I$ are deterministic volatilities and $W^E$, $W^I$ are correlated Brownian motions with a correlation parameter $\rho\in(-1,1)$. The process $F^E$ is the option price of a future contact written on some energy price and $F^I$ is the option price of a future contact written on some temperature price.\\
Given an arbitrary $W^E$, there exists $\widetilde{W}^I$ which is independent of $W^E$ and $W^I$. Then, we can express $W^I$ as follows
\begin{equation}
W^I=\rho W^E+\sqrt{1-\rho^2}\widetilde{W}^I.
\end{equation}
Thus we have
\begin{eqnarray}\label{Gen HJB}
 \quad\quad dF^E(t;\tau_1,\tau_2) &=&\sigma_E(t, F^E(t;\tau_1,\tau_2))dW^E(t),\\
 \quad\quad dF^I(t;\tau_1,\tau_2) &=& \rho\sigma_I(t,F^I(t;\tau_1,\tau_2))dW^E(t)+\sigma_I(t,F^I(t;\tau_1,\tau_2))\sqrt{1-\rho^2} d\widetilde{W}^I(t)\,.
\end{eqnarray}
The above equations can be written in matrix form as follows:
$$
\left(
  \begin{array}{c}
    dF^E \\
    dF^I \\
  \end{array}
\right)
=
\left(
  \begin{array}{cc}
   \sigma_E(t, F^E)& 0 \\
 \rho  \sigma_I(t, F^I) & \sigma_I(t,F^I)\sqrt{1-\rho^2} \\
  \end{array}
\right)
\left(
  \begin{array}{c}
    dW^E\\
    d\widetilde{W}^I \\
  \end{array}
\right)\,.
$$
We can write this as:
\begin{equation}\label{HJBvector}
 d\bar{F}=a(t,F^I,F^E)d\bar{W}\,,
\end{equation}
where the matrix $a:([0,\tau_2]\times\mathbb{R}^2)\rightarrow \mathcal{M}_2$,
satisfies the growth and Lipschitz conditions. We can write \eqref{HJBvector} as:
\begin{equation}\label{HJBInt}
  \bar{F}=\bar{F_0}+\int_{0}^{t}a(t,F^I,F^E)d\bar{W}\,,~~~\bar{F_0}>\textbf{0}\,.
\end{equation}
Given this dynamics the quanto option becomes :
\begin{equation}\label{generoption}
 C_t=\mathbb E^{\mathbb Q}\left[\tilde{g}\left(\int_{0}^{\tau_2}\sigma_E(t,F^E)dW^E\right)\tilde{h}\left(\int_{0}^{\tau_2}\sigma_I(t,F^I)dW^E,
\int_{0}^{\tau_2}\sigma_I(t,F^I)d\widetilde{W}^I\right)\right]\,,
\end{equation}
where $\tilde{g}(x)=(x-K^E)^+$ and $\tilde{h}(x,y)=(\rho
x+\sqrt{1-\rho^2}y-K^I)^+$ are measurable functions.
\section{A Primer on the Malliavin Derivative Properties}
In this section, we review the necessary Malliavin derivative properties.
These properties were also highlighted in Fourni\'{e} \emph{et al.} \cite{Fournie1} and Mhlanga \cite{Mhlanga} and the proofs can be found in Nualart
\cite{Nualart}.
Let $\{W(t),~0\leq t\leq\tau_2\}$ be an $n$-dimensional Brownian motion defined on a
complete probability space $(\Omega, \mathcal{F},\mathbb{F},\mathbb Q)$.
Let $S$ denote the class of random variables of the form
\begin{equation*}
  F=f\left(\int_{0}^{\tau_2}h_1(t)dW(t),\cdots,\int_{0}^{\tau_2}h_n(t)dW(t)\right)\,,\quad
f\in C^{\infty}(\mathbb{R}^n)\,,
\end{equation*}
where $h_1,\cdots, h_n\in L^2([0,\tau_2])$.

 For $F\in S$, the Malliavin derivative $DF$ of $F$ is defined as the process
$\{D_tF,~t\in[0,\tau_2]\}$ in $L^2([0,\tau_2])$ is defined by :
\begin{equation*}
  D_tF=\sum_{i=1}^{n}\frac{\partial f}{\partial
x_i}\left(\int_{0}^{\tau_2}h_1(t)dW(t),\cdots,\int_{0}^{\tau_2}h_n(t)dW(t)\right)h_i(t)\,,
\quad t\geq 0~~a.s.
\end{equation*}
On $L^2([0,\tau_2])$ define the norm as :
\begin{equation*}
 ||F||_{1,2}:=\left(\mathbb{E}^{\mathbb Q}|F|^2+\mathbb{E}^{\mathbb Q}[\int_0^{\tau_2}|D_tF|^2dt]\right)^{
\frac12}\,.
\end{equation*}
The chain rule holds for the Malliavin derivative in the following form.
\subsubsection*{Property P1}
 Let $F=(F_1,\ldots,F_n)\in \mathbb D^{1,2}$ and let $\varphi :
\mathbb R^n\rightarrow \mathbb R$ be a continuously differentiable
function with bounded partial derivatives. Then $\varphi(F)\in
{\mathbb D}^{1,2}$ and
\begin{equation}\label{6}
D_t\varphi(F)=\sum_{i=1}^n\frac{\partial \varphi}{\partial
x_i}(F)D_tF_i,~~t\geq0~~~~~a.s.
\end{equation}

\subsubsection*{Property P2}
Let $\{X_t,~t\geq0\}$ be an $\mathbb R^n$ valued It$\hat{\text{o}}$ process whose dynamics are governed by the stochastic differential equation
\begin{equation}
dX_t=b(X_t)dt+\sigma(X_t)dW_t,
\end{equation}
where $b$ and $\sigma$ are supposed to be continuously differentiable functionals with bounded derivatives and $\sigma(x)\neq0$ for all $x\in\mathbb R^n$. Let $\{Y_t,~t\geq0\}$ be the associated first variation process given by the stochastic differential equation
\begin{equation}
dY_t=b^{\prime}(X_t)Y_tdt+\sum_{i=1}^n\sigma_i^{\prime}(X_t)Y_tdW_t^i,~~~~Y_0=I_n,
\end{equation}
where $I_n$ is the identity matrix of $\mathbb R^n$, primes denote derivatives and $\sigma_i$ is the $i$-th column vector of $\sigma$. The the process $\{X_t,~t\geq0\}$ belongs to $\mathbb D^{1,2}$ and its Malliavin derivative is given by
 \begin{equation}
D_rX_t=Y_tY_r^{-1}\sigma(X_r)1_{\{r\leq t\}}, ~~r\geq0~~a.s.,
\end{equation}
which is equivalent to
\begin{equation}
Y_t=D_rX_t\sigma^{-1}(X_r)Y_r1_{\{r\leq t\}}~~~~~~a.s.
\end{equation}
The Malliavin derivative has an adjoint operator called Skorohod integral (also
known as the divergence operator $\delta$). We shall denote the domain
of the adjoint operator $\delta$ by Dom($\delta$).
\subsubsection*{Property P3}
Let $u\in L^2(\Omega\times [0,\tau_2])$. Then $u$ belongs to the domain
 $\text{Dom}(\delta)$ of $\delta$
 if for all $F\in \mathbb D^{1,2}$ we have
  \begin{equation}\mid  \mathbb E\left[\langle{D}F,u\rangle_{L^2(\Omega)}\right]\mid=
   \mid \mathbb E\left[\int_0^{\tau_2}D_tFu(t)dt\right]\mid\leq c\parallel
F\parallel_{L^2(\Omega)}\end{equation}
 where $c$
is some constant depending on $u$.
 If $u$ belongs to Dom$(\delta)$, then
\begin{equation}\delta(u)=\int_0^{\tau_2}u_t\delta
 W_t\end{equation} is the element of $L^2(\Omega)$ such that the integration by parts
 formula holds:
\begin{equation}\label{8}
  \mathbb E\left[\left(\int_0^{\tau_2} {D}_tF u_tdt\right)\right]=\mathbb E[F\delta(u)]~
  ~~\text{for all}~~~F\in\mathbb D^{1,2}.
 \end{equation}

 An important property of the Skorohod integral $\delta$ is that its
domain Dom$(\delta)$ contains all adapted stochastic processes
which belong to $L^2(\Omega\times [0,\tau_2])$. For such processes the
Skorohod integral $\delta$ coincides with the It$\hat{\text{o}}$
stochastic integral.
\subsubsection*{Property P4}
 If $u$ is an adapted process belonging to
$L^2(\Omega\times[0,\tau_2])$, then
\begin{equation}\label{ds12}\delta(u)=\int_0^{\tau_2}u(t)dW_t.\end{equation}

Further, if the random variable $F$ is ${\mathcal F}_{\tau_2}$-adapted and
belongs to $\mathbb D^{1,2}$ then, for any $u$ in $Dom(\delta)$,
the random variable $Fu$ will be Skorohod integrable.
\subsubsection*{Property P5}
 Let $F$ belongs to $\mathbb D^{1,2}$ and $u\in Dom(\delta)$ such
that $ \mathbb E[\int_0^{\tau_2}F^2u_t^2dt]<\infty$. Then
$Fu\in{Dom}(\delta)$ and
 \begin{equation}
\delta(Fu)=F\delta(u)-\int_0^{\tau_2}{D}_tFu_tdt,
\end{equation}
whenever the right hand side belongs to $L^2(\Omega)$. In
particular, if $u$ is moreover adapted, we have
\begin{equation}\label{9}
\delta(Fu)=F\int_0^{\tau_2}u_tdW_t-\int_0^{\tau_2}D_tFu_tdt.
\end{equation}

%
%

\section{The Independent Case}
From the diffusion stochastic differential equation \eqref{Gen HJB} with
$\rho=0$, consider the following HJM
risk-neutral dynamics of the forward contract at time $t$. We call this the `independent case'.
Let the future price processes under the risk-neutral measure $\mathbb{Q}$ be
given as
\begin{equation*}
dF^i(t;\tau_1,\tau_2)=\sigma_i(t;\tau_1,\tau_2)F^i(t;\tau_1,\tau_2)dW^i(\tau),
~~F^i(0;\tau_1,\tau_2)>0,
\end{equation*}
for $E,I=i$. The function $F^i(0;\tau_1,\tau_2)$ represents today's forward
price. We call this, the independent case since $\rho=0$.  Explicitly this can
be written as:
\begin{eqnarray*}
F^E(\tau_2;\tau_1,\tau_2)&=&F^E(0;\tau_1,\tau_2)\exp\left(-\frac12\int_0^{\tau_2
}\sigma_E^2(u;\tau_1,\tau_2)du+\int_0^{\tau_2}\sigma_E(u;\tau_1,
\tau_2)dW^E(u)\right)\\
F^I(\tau_2;\tau_1,\tau_2)&=&F^I(0;\tau_1,\tau_2)\exp\left(-\frac12\int_0^{\tau_2
}\sigma_I^2(u;\tau_1,\tau_2)du+\int_0^{\tau_2}\sigma_I(u;\tau_1,
\tau_2)dW^I(u)\right)\,,\\
\end{eqnarray*}
where $W^E$ and $W^I$ are Brownian motions. Here,
$\int_0^{\tau_2}\sigma_i^2(\tau;\tau_1,\tau_2)d\tau<\infty$ meaning $\tau\mapsto
F^i(\tau;\tau_1,\tau_2)$ is a martingale.
Introduce $g:\mathbb{R}\mapsto \mathbb{R}$ and $h:\mathbb{R}\mapsto \mathbb{R}$
measurable functions. The payoff structure of a quanto option on the forwards
with maturity at time $\tau_2$ given by
\begin{equation}
\label{generoption}
 C=\mathbb{E}^{\mathbb Q}[ g(F^E(\tau_2;\tau_1,\tau_2 ))h(F^I(\tau_2;\tau_1,\tau_2 ))]\,,
\end{equation}
where $g(x)=(x-K^E)^+$ and $h(x)=(x-K^I)^+$ and the risk-free interest rate
$r=0$.
We assume the following integrability conditions:
$$
\mathbb{E}[g^2(F^E(\tau_2;\tau_1,\tau_2))]<\infty,\quad
\mathbb{E}[h^2(F^I(\tau_2;\tau_1,\tau_2))]<\infty\,.
$$
At several places, we
will require the diffusion matrix $\sigma_i$, $i=E, I$ to satisfy the following
condition:
\begin{equation}\label{d1117} \exists
\eta>0~~~~\xi^*\sigma^*(t;\tau_1,\tau_2)\sigma(t;\tau_1,\tau_2)\xi>\eta\mid\xi\mid^2~~~\text{for
all}~~ \xi\in\mathbb R^n,~~t\in[\tau_1,\tau_2]~~\text{with}~~\xi\neq0. \end{equation}
where $\xi^*$ denotes the transpose of $\xi$. This is called the
\textit{uniform ellipticity condition}.\\
 The weight function obtained when computing Greeks using the integration by parts formula
 should not degenerate with probability one, otherwise the computation
 will not be valid. To avoid this degeneracy we introduce the set $\Upsilon_n$ (see \cite{Fournie1}) defined by
\begin{equation}
\Upsilon_n=\{a\in L^2([0,\tau_2])\mid\int_0^{t_i}a(t)dt=1~~\text {for
all}~~ i=1,\ldots,n\}.
\end{equation}
\begin{prop}
\label{DeltaE}
Assume that the diffusion matrix $\sigma_E$ is uniformly elliptic. Then for all $a\in\Upsilon_n$, the delta of the energy option is given by :
 \begin{equation*}
 \Delta_E=\mathbb{E}^{\mathbb Q}[g(F^E(\tau_2;\tau_1,
\tau_2))h(F^I(\tau_2;\tau,\tau_2))\pi^{\Delta_E}],
\end{equation*}
where the Malliavin weight $\pi^{\Delta_E}$ is
$$\pi^{\Delta_E}=\int_0^{\tau_2}a(t)\left(\sigma_E^{-1}(t;\tau_1,\tau_2)Y_E(t;\tau_1,\tau_2)\right)^*dW^E(t).$$
\end{prop}
 \begin{proof}
Let $g$ be a continuously differentiable function with bounded derivatives. Introduce
$$Y_E(t;\tau_1,\tau_2)=\exp\left(-\frac12\int_0^t\sigma^2_E(u;\tau_1,
\tau_2)du+\int_0^t\sigma_E(u;\tau_1,\tau_2)dW^E(u)\right).$$ This implies that
$$F^E(\tau;\tau_1,\tau_2)=F^E(0;\tau_1,\tau_2)Y_E(t;\tau_1,\tau_2).$$ An application of \emph{Property $P2$} shows that $F^E(\tau_2;\tau_1,\tau_2)$ belongs to $\mathbb D^{1,2}$ and we have:
\begin{eqnarray*}
D_tF^E(\tau_2;\tau_1,\tau_2)&=&Y_E(\tau_2;\tau_1,
\tau_2)Y_E^{-1}(t;\tau_1,
\tau_2)\sigma_E(t;\tau_1,\tau_2)1_{t<\tau_2}.
\end{eqnarray*}
This is equivalent to
\begin{eqnarray*}
Y_E(\tau_2;\tau_1,\tau_2)1_{t<\tau_2}&=&D_tF^E(\tau_2;\tau_1,\tau_2)\sigma_E^{-1}
(t;\tau_1,\tau_2)Y_E(t;\tau_1,\tau_2).
\end{eqnarray*}
Multiply both sides by a square function which integrates to 1 on $[0,\tau_2]$
\begin{equation*}
Y_E(\tau_2;\tau_1,\tau_2)=\int_0^{\tau_2}
D_tF^E(\tau_2;\tau_1,\tau_2)a(t)\sigma_E^{-1}(t;\tau_1,\tau_2)Y_E(t;\tau_1,\tau_2)dt.
\end{equation*}
Now
\begin{eqnarray*}
\Delta_E&:=&\frac{\partial}{\partial F^E(0;\tau_1,\tau_2)}\mathbb E^Q[g(F^E(\tau_2;\tau_1,
\tau_2))h(F^I(\tau_2;\tau_1,\tau_2))]\\
&=&\mathbb E^Q[g^{\prime}(F^E(\tau_2;\tau_1,
\tau_2))h(F^I(\tau_2;\tau_1,\tau_2))\frac{\partial F^E(\tau_2;\tau_1,
\tau_2)}{\partial F^E(0;\tau_1,
\tau_2)}]\\
&=&\mathbb{E}^{\mathbb Q}[g^{\prime}(F^E(\tau_2;\tau_1,
\tau_2))h(F^I(\tau_2;\tau_1,\tau_2))Y_E(\tau_2;\tau_1,\tau_2)]\\
&=&\mathbb{E}^{\mathbb Q}[\int_0^{\tau_2}g^{\prime}(F^E(\tau_2;\tau_1,
\tau_2))h(F^I(\tau_2;\tau_1,\tau_2))\\
&&\quad\times D_tF^E(\tau_2;\tau_1,\tau_2)a(t)\sigma_E^{-1}(t;\tau_1,\tau_2)Y_E(t;\tau_1,\tau_2)dt]\\
&=&\mathbb{E}^{\mathbb Q}
\left[h(F^I(\tau_2;\tau_1,\tau_2))\int_0^{\tau_2}D_tg(F^E(\tau_2;\tau_1,\tau_2))a(t)\sigma_E^{-1}(t;\tau_1,\tau_2)Y_E(t;\tau_1,\tau_2)dt\right]\\
&=&\mathbb{E}^{\mathbb Q}\left[g(F^E(\tau_2;\tau_1,
\tau_2))h(F^I(\tau_2;\tau_1,\tau_2))\int_0^{\tau_2}a(t)\left(\sigma_E^{-1}(t;\tau_1,
\tau_2)Y_E(t;\tau_1,\tau_2)\right)^*dW^E(t)\right]\,\\
\end{eqnarray*}
where $g^{\prime}$ denotes the derivative of $g$ with respect to $F^E(0;\tau_1,\tau_2)$. Here, we have used the chain rule property, (\emph{Property $P1$}), the integration by parts formula (\emph{Property $P3$}), and the fact that the Skorohod integral coincides with the It$\hat{\text{o}}$ stochastic integral (\emph{Property $P4$}).\\ Since a continuously differentiable function is dense in $L^2$, the result hold for any $g\in L^2$ (see Fourni$\acute{\text{e}}$ et al. \cite{Fournie1} for details).
\end{proof}
Similarly, we obtain the following result.
\begin{prop}
\label{DeltaI}
Assume that the diffusion matrix $\sigma_I$ is uniformly elliptic. Then for all $a\in\Upsilon_n$, the delta of the energy temperature is given
by:
 \begin{equation*}
 \Delta_I=\mathbb{E}^{\mathbb Q}[g(F^E(\tau_2;\tau_1,
\tau_2))h(F^I(\tau_2;\tau,\tau_2))\pi^{\Delta_I}],
\end{equation*}
where the Malliavin weight $\pi^{\Delta_I}$ is
$$\pi^{\Delta_I}=\int_0^{\tau_2}a(t)\left(\sigma_I^{-1}(t;\tau_1,\tau_2)Y_I(t;\tau_1,\tau_2)\right)^*dW^I(t).$$
\end{prop}
\begin{proof}
Follows along the lines of the proof of Proposition of \ref{DeltaE}.
\end{proof}
The following result gives the cross-gamma hedge in the independent case:
\begin{prop}\label{DeltaEI}
Assume that the diffusion matrices $\sigma_i$, $i=E, I$ are uniformly elliptic. Then for all $a\in\Upsilon_n$, the following hold:
 \begin{equation*}
 \Delta_{EI}=\mathbb{E}^{\mathbb Q}[
g(F^E(\tau_2;\tau_1,\tau_2))h(F^I(\tau_2;\tau,\tau_2))\pi^{\Delta_{EI}}],
\end{equation*}
where the Malliavin weight $\pi^{\Delta_{EI}}$ is
\begin{equation*}
 \pi^{\Delta_{EI}}=\int_0^{\tau_2}a(t)\left(\sigma_E^{-1}(t;\tau_1,\tau_2)Y_E(t;\tau_1,\tau_2)\right)^*dW^E(t)\int_0^{\tau_2}a(t)\left(\sigma_I^{-1}(t;\tau_1,\tau_2)Y_I(t;\tau_1,\tau_2)\right)^*dW^I(t).
\end{equation*}
 \end{prop}
\begin{proof}
We first assume that $g$ and $h$ are continuously differentiable with bounded derivatives. From Proposition \ref{DeltaE} we have
\begin{equation*}
 \Delta_E=\mathbb{E}^{\mathbb Q}[g(F^E(\tau_2;\tau_1,
\tau_2))h(F^I(\tau_2;\tau,\tau_2))Z^E(\tau_2)].
\end{equation*}
As in Proposition \ref{DeltaE}, we introduce
$$Y_I(t;\tau_1,\tau_2)=\exp\left(-\frac12\int_0^t\sigma^2_I(u;\tau_1,
\tau_2)du+\int_0^t\sigma_I(u;\tau_1,\tau_2)dW^I(u)\right).$$ This implies that
$$F^I(\tau;\tau_1,\tau_2)=F^I(0;\tau_1,\tau_2)Y_I(t;\tau_1,\tau_2).$$ An application of \emph{Property $P2$} shows that $F^I(\tau_2;\tau_1,\tau_2)$ belongs to $\mathbb D^{1,2}$ and we have:
\begin{eqnarray*}
D_tF^I(\tau_2;\tau_1,\tau_2)&=&Y_I(\tau_2;\tau_1,\tau_2)Y_I^{-1}(t;\tau_1,\tau_2)\sigma_I(t;\tau_1,\tau_2)1_{t<\tau_2}.
\end{eqnarray*}
This is equivalent to
\begin{eqnarray*}
Y_I(\tau_2;\tau_1,\tau_2)1_{t<\tau_2}&=&D_tF^I(\tau_2;\tau_1,\tau_2)\sigma_I^{-1}(t;\tau_1,\tau_2)Y_I(t;\tau_1,\tau_2).
\end{eqnarray*}
Multiply both sides by a square function which integrates to 1 on $[0,\tau_2]$
\begin{equation*}
Y_I(\tau_2;\tau_1,\tau_2)=\int_0^{\tau_2}
D_tF^I(\tau_2;\tau_1,\tau_2)a(t)\sigma_I^{-1}(t;\tau_1,\tau_2)Y_I(t;\tau_1,\tau_2)dt.
\end{equation*}
Now
\begin{eqnarray*}
\Delta_{EI}&:=&\frac{\partial}{\partial F^I(0;\tau_1,\tau_2)}\Delta_E\\
&=&\frac{\partial}{\partial F^I(0;\tau_1,\tau_2)}\left[ \mathbb{E}^{\mathbb Q}[g(F^E(\tau_2;\tau_1,
\tau_2))h(F^I(\tau_2;\tau,\tau_2))Z^E(\tau_2)] \right]\\
&=&\mathbb{E}^{\mathbb Q}[g(F^E(\tau_2;\tau_1,
\tau_2))Z^E(\tau_2)h^{\prime}(F^I(\tau_2;\tau,\tau_2))\frac{\partial F^I(\tau_2;\tau_1,\tau_2)}{\partial F^I(0;\tau_1,\tau_2)}]\\
&=&\mathbb{E}^{\mathbb Q}[g(F^E(\tau_2;\tau_1,
\tau_2))Z^E(\tau_2)h^{\prime}(F^I(\tau_2;\tau,\tau_2))Y_I(\tau_2;\tau_1,\tau_2)]\\
&=&\mathbb{E}^{\mathbb Q}[\int_0^{\tau_2
}g(F^E(\tau_2;\tau_1,\tau_2))h^{\prime}(F^I(\tau_2;\tau_1,\tau_2))Z^E(\tau_2)\\
&&\quad\times
D_tF^I(\tau_2;\tau_1,\tau_2)a(t)\sigma_I^{-1}(t;\tau_1,\tau_2)Y_I(t;\tau_1,\tau_2)dt]\\
&=&\mathbb{E}^{\mathbb Q}[
g(F^E(\tau_2;\tau_1,\tau_2))Z^E(\tau_2)\\
&&\quad\times\int_0^{\tau_2}D_t(h(F^I(\tau_2;\tau_1,\tau_2)))a(t)\sigma_I^{-1}(t;\tau_1,
\tau_2)Y_I(t;\tau_1,\tau_2)dt]\\
&=&\mathbb{E}^{\mathbb Q}[
g(F^E(\tau_2;\tau_1,\tau_2))h(F^I(\tau_2;\tau_1,\tau_2))Z^E(\tau_2)\\
&&\quad\times\int_0^{\tau_2}a(t)\left(\sigma_I^{-1}(t;\tau_1,\tau_2)Y_I(t;\tau_1,\tau_2)\right)^*dW^I(t)]\,.\\
\end{eqnarray*}
Here, we have used the chain rule property, (\emph{Property $P1$}), the integration by parts formula (\emph{Property $P3$}), and the fact that the Skorohod integral coincides with the It$\hat{\text{o}}$ stochastic integral (\emph{Property $P4$}).\\
 The result can be extended to the general case by a density argument. We omit the details.
\end{proof}
\section{The Correlation Case}
We consider the following HJM
\begin{eqnarray}
 \quad\quad dF^E(t;\tau_1,\tau_2) &=&\sigma_E(t, F^E(t;\tau_1,\tau_2))dW^E(t),\\
 \quad\quad dF^I(t;\tau_1,\tau_2) &=& \rho\sigma_I(t,F^I(t;\tau_1,\tau_2))dW^E(t)+\sigma_I(t,F^I(t;\tau_1,\tau_2))\sqrt{1-\rho^2} d\widetilde{W}^I(t)\,.
\end{eqnarray}
with $F^E(0;\tau_1,\tau_2)>0,$ and $F^I(0;\cdot)>0.$\\
  That is, we consider the case when there is correlation between $F^E$ and $F^I$.
Suppose the Brownian motions $B_1$ and $B_2$ are independent. Let $W_1=B_1$ and
$W_2=\rho B_1+\sqrt{1-\rho^2}B_2$. This implies that
$$
g(W_1)h(W_2)=g(B_1)h( \rho B_1+\sqrt{1-\rho^2}B_2).
$$
In this setting, we have the following quanto option structure:
$$
C=\mathbb{E}^{\mathbb Q}[g(F^E(\tau_2;\tau_1,\tau_2))h(\rho
F^E(\tau_2;\tau_1,\tau_2)+\sqrt{1-\rho^2} F^I(\tau_2;\tau_1,\tau_2))]\,.
$$
Now we derive the energy delta.
\begin{prop}
\label{DeltaTempcor}
Assume that the diffusion matrix $\sigma_E$ is uniformly elliptic. Then for all $a\in\Upsilon_n$, the following hold:
\begin{eqnarray*}
\Delta_E&=&\mathbb{E}^{\mathbb Q}[g(F^E(\tau_2;\tau_1,\tau_2))h(\rho F^E(\tau_2;\tau_1,\tau_2)+\sqrt{1-\rho^2}F^I(\tau_2;\tau_1,\tau_2))\pi^{\Delta_E}  (1+\rho)]\,,\\
\end{eqnarray*}
where the Malliavin weight $\pi^{\Delta_E}$ is
$$\pi^{\Delta_E}=\int_0^{\tau_2}a(t)\left(\sigma_E^{-1}(t;\tau_1,\tau_2)Y_E(t;\tau_1,\tau_2)\right)^*dW^E(t).$$
\end{prop}

\begin{proof}
Let $g$ be a continuously differentiable function with bounded derivatives. As in Proposition \ref{DeltaE}, introduce
$$Y_E(t;\tau_1,\tau_2)=\exp(-\frac12\int_0^t\sigma^2_E(u;\tau_1,
\tau_2)du+\int_0^t\sigma_E(u;\tau_1,\tau_2)dW^E(u)).$$ This implies that
$$F^E(t;\tau_1,\tau_2)=F^E(0;\tau_1,\tau_2)Y_E(t;\tau_1,\tau_2).$$
An application of \emph{Property $P2$} shows that $F^E(\tau_2;\tau_1,\tau_2)$ belongs to $\mathbb D^{1,2}$ and we have:
\begin{eqnarray*}
D_tF^E(\tau_2;\tau_1,\tau_2)&=&Y_E(\tau_2;\tau_1,
\tau_2)Y_E^{-1}(t;\tau_1,\tau_2)\sigma_E(t;\tau_1,\tau_2)1_{t<\tau_2}.
\end{eqnarray*}
This is equivalent to
\begin{eqnarray*}
Y_E(\tau_2;\tau_1,\tau_2)1_{t<\tau_2}&=&D_tF^E(\tau_2;\tau_1,\tau_2)\sigma_E^{-1}
(t;\tau_1,\tau_2)Y_E(t;\tau_1,\tau_2).
\end{eqnarray*}
Multiply both sides by a square function which integrates to 1 on $[0,\tau_2]$
\begin{equation*}
Y_E(\tau_2;\tau_1,\tau_2)=\int_0^{\tau_2}
D_tF^E(\tau_2;\tau_1,\tau_2)a(t)\sigma_E^{-1}(t;\tau_1,\tau_2)Y_E(t;\tau_1,\tau_2)dt.
\end{equation*}
Now
\begin{eqnarray*}
\Delta_E&=&\mathbb{E}^{\mathbb Q}[g^{\prime}(F^E(\tau_2;\tau_1,\tau_2))h(\rho
F^E(\tau_2;\tau_1,\tau_2)+\sqrt{1-\rho^2}
F^I(\tau_2;\tau_1,\tau_2))Y_E(\tau_2;\tau_1,\tau_2)\\
&&\quad +g(F^E(\tau_2;\tau_1,\tau_2))h^{\prime}(\rho
F^E(\tau_2;\tau_1,\tau_2)+\sqrt{1-\rho^2} F^I(\tau_2;\tau_1,\tau_2))\rho
Y_E(\tau_2;\tau_1,\tau_2)]\\
&=&\mathbb{E}^{\mathbb Q}[\int_0^{\tau_2}g^{\prime}(F^E(\tau_2;\tau_1,
\tau_2))h(\rho F^E(\tau_2;\tau_1,\tau_2)+\sqrt{1-\rho^2}
F^I(\tau_2;\tau_1,\tau_2))\\
&&\quad \times
D_tF^E(\tau_2;\tau_1,\tau_2)a(t)\sigma_E^{-1}(t;\tau_1,\tau_2)Y_E(t;\tau_1,\tau_2)dt\\
&&\quad +\int_0^{\tau_2}
g(F^E(\tau_2;\tau_1,\tau_2))h^{\prime}(\rho F^E(\tau_2;\tau_1,\tau_2)+\sqrt{1-\rho^2}
F^I(\tau_2;\tau_1,\tau_2))\\
&&\quad\times D_tF^E(\tau_2;\tau_1,\tau_2)a(t)\sigma_E^{-1}(t;\tau_1
,\tau_2)Y_E(t;\tau_1,\tau_2)dt]\\
&=&\mathbb{E}^{\mathbb Q}[h(\rho
F^E(\tau_2;\tau_1,\tau_2)+\sqrt{1-\rho^2} F^I(\tau_2;\tau_1,\tau_2))\\
&&\times\int_0^{\tau_2}D_tg(F^E(\tau_2;\tau_1,\tau_2))a(t)\sigma_E^{-1}(t;\tau_1
,\tau_2)Y_E(t;\tau_1,\tau_2)dt
 +\rho g(F^E(\tau_2;\tau_1,\tau_2))\\
&&\quad \times\int_0^{\tau_2}D_t h(\rho
F^E(\tau_2;\tau_1,\tau_2)+\sqrt{1-\rho^2}
F^I(\tau_2;\tau_1,\tau_2))a(t)\sigma_E^{-1}(t;\tau_1,\tau_2)Y_E(t;\tau_1,\tau_2)dt]\\
&=&\mathbb{E}^{\mathbb Q}[h(\rho
F^E(\tau_2;\tau_1,\tau_2)+\sqrt{1-\rho^2} F^I(\tau_2;\tau_1,\tau_2))\\
&&\quad\times
g(F^E(\tau_2;\tau_1,\tau_2))\int_0^{\tau_2}a(t)\left(\sigma_E^{-1}(t;\tau_1,
\tau_2)Y_E(t;\tau_1,\tau_2)\right)^*dW^E(t)\\
&&\quad +\rho g(F^E(\tau_2;\tau_1,\tau_2))h(\rho F^E(\tau_2;\tau_1,\tau_2)+\sqrt{1-\rho^2}
F^I(\tau_2;\tau_1,\tau_2))\\
&&\quad \times\int_0^{\tau_2}a(t)\left(\sigma_E^{-1}(t;\tau_1,\tau_2)Y_E(t;\tau_1,\tau_2)\right)^*dW^E(t)]\,.\\
\end{eqnarray*}
As in the proof of Proposition \ref{DeltaE}, we have used the chain rule property, (\emph{Property $P1$}), the integration by parts formula (\emph{Property $P3$}), and the fact that the Skorohod integral coincides with the It$\hat{\text{o}}$ stochastic integral (\emph{Property $P4$}).\\
The result can be extended to the general case by a density argument. We omit the details.
\end{proof}
Now we derive the temperature delta in the correlation case.
\begin{prop}
Assume that the diffusion matrix $\sigma_I$ is uniformly elliptic. Then for all $a\in\Upsilon_n$, the following hold :
 \begin{eqnarray*}
\Delta_I&=&\sqrt{1-\rho^2}\mathbb{E}^{\mathbb Q}[g(F^E(\tau_2;\tau_1,\tau_2))h(\rho
F^E(\tau_2;\tau_1,\tau_2)+\sqrt{1-\rho^2} F^I(\tau_2;\tau_1,\tau_2)) \pi^{\Delta_I}].
 \end{eqnarray*}
where the Malliavin weight $\pi^{\Delta_I}$ is
$$\pi^{\Delta_I}=\int_0^{\tau_2}a(t)\left(\sigma_I^{-1}(t;\tau_1,\tau_2)Y_I(t;\tau_1,\tau_2)\right)^*dW^I(t).$$
\end{prop}
\begin{proof}
The proof is similar to that of Proposition \ref{DeltaTempcor}.
\end{proof}
The following result gives the cross-gamma hedge in the correlation case.
\begin{prop}
Assume that the diffusion matrices $\sigma_i$, $i=E, I$, are uniformly elliptic. Then for all $a\in\Upsilon_n$, the following hold:
 \begin{eqnarray*}
\Delta_{EI}&=&\sqrt{1-\rho^2}\mathbb{E}^{\mathbb Q}[g(F^E(\tau_2;\tau_1,\tau_2))h(\rho F^E(\tau_2;\tau_1,\tau_2)+\sqrt{1-\rho^2}
F^I(\tau_2;\tau_1,\tau_2))
 \pi^{\Delta_{EI}}\\
&&~+\rho\sqrt{1-\rho^2}
g(F^E(\tau_2;\tau_1,\tau_2))h(\rho F^E(\tau_2;\tau_1,\tau_2)+\sqrt{1-\rho^2}
F^I(\tau_2;\tau_1,\tau_2))\pi^{\Delta_{EI}}]\,\\
\end{eqnarray*}
where the Malliavin weight $\pi^{\Delta_{EI}}$ is:
\begin{equation*}
\pi^{\Delta_{EI}}=\int_0^{\tau_2}a(t)\left(\sigma_E^{-1}(t;\tau_1,\tau_2)Y_E(t;\tau_1,\tau_2)\right)^*dW^E(t)\int_0^{\tau_2}a(t)\left(\sigma_I^{-1}(t;\tau_1,\tau_2)Y_I(t;\tau_1,\tau_2)\right)^*dW^I(t).
\end{equation*}
\end{prop}
\begin{proof}
The proof follows the same line of argument as in the proof of Proposition
\ref{DeltaEI}. The details are omitted.
\end{proof}
\subsection{The Residual Risk}
If we take the independent delta of energy $\Delta^{Ind}$, say, as the
benchmark value and the correlated case as $\Delta^{Corr}$. Then the residual
risk is determined by the difference between the independent delta of energy and
the correlated case as follows:
$$
|\Delta^{Corr}-\Delta^{Ind}|,
$$
for each $\rho$. The same analysis goes for the cross-gamma formulae.
\section{Examples}
We will provide Malliavin weights in the case where the quanto option payoff functions depend on the terminal value, that is, $\tau_2=T$
\subsection{The independent case}
We consider the following stochastic differential equations to describe the energy price $F^E$ and the temperature price $F^I$ dynamics
\begin{eqnarray}\label{ABC}
\frac{dF^E}{F^E}&=&\sigma_EdW_t^E,~~~~F^E(0)>0\\
\frac{dF^I}{F^I}&=&\sigma_IdW_t^I,~~~~F^I(0)>0,
\end{eqnarray}
where $\sigma_E$, $\sigma_I$ are deterministic volatilities and $W^E$, $W^I$ are independent Brownian motions.
The quanto option pricing formula is then expressed as
\begin{equation}
C_t=\mathbb E^{\mathbb Q}[g(F^E)h(F^I)].
\end{equation}
By using the general formulae developed in the previous sections, we are able to compute analytically the values of different Malliavin weights. Here we set $a(t)=\frac{1}{T}$. We have
\begin{eqnarray*}
\pi^{\Delta_E}&=&\frac{1}{F^E(0)T}\int_0^T\frac{1}{\sigma_E}dW^E(t).
\end{eqnarray*}

\begin{eqnarray*}
\pi^{\Delta_I}&=&\frac{1}{F^I(0)T}\int_0^T\frac{1}{\sigma_I}dW^I(t).
\end{eqnarray*}
\begin{eqnarray*}
\pi^{\Delta_{EI}}&=&\frac{1}{F^E(0)F^I(0)T^2}\left(\int_0^T\frac{1}{\sigma_E}dW^E(t)\right)\left(\int_0^T\frac{1}{\sigma_I}dW^I(t)\right).
\end{eqnarray*}
\subsection{The correlation case}
Again, we consider the following stochastic differential equations to describe the energy price $F^E$ and the temperature price $F^I$ dynamics
\begin{eqnarray}\label{ABCd}
\frac{dF^E}{F^E}&=&\sigma_EdW_t^E,~~~~F^E(0)>0\\
\frac{dF^I}{F^I}&=&\rho\sigma_IdW_t^I+\sigma_I\sqrt{1-\rho^2}d\widetilde{W}^I,~~~~F^I(0)>0,
\end{eqnarray}
 where $W_t^E$, $W_t^I$ are correlated Brownian motions with correlation parameter $\rho\in(-1,1)$.
The system of stochastic differential equations can be written in a matrix form
$$
\left(
  \begin{array}{c}
    dF^E \\
    dF^I \\
  \end{array}
\right)
=\left(
   \begin{array}{cc}
     \sigma_EF^E & 0 \\
     \rho\sigma_IF^I & \sigma_I\sqrt{1-\rho^2}F^I \\
   \end{array}
 \right)
 \left(
   \begin{array}{c}
     dW^E \\
     d\widetilde{W}^I\\
   \end{array}
 \right)\,.
 $$
 The inverse matrix of
 $$
\left(
   \begin{array}{cc}
     \sigma_EF^E & 0 \\
     \rho\sigma_IF^I & \sigma_I\sqrt{1-\rho^2}F^I \\
   \end{array}
  \right)\,
 $$
 is calculated as
  $$
\frac{1}{\sigma_E\sigma_I\sqrt{1-\rho^2}F^EF^I}\left(
   \begin{array}{cc}
     \sigma_I\sqrt{1-\rho^2}F^I & 0 \\
     -\rho\sigma_IF^I & \sigma_EF^E \\
   \end{array}
  \right)\,=\left(
   \begin{array}{cc}
     \frac{1}{\sigma_EF^E} & 0 \\
     -\frac{\rho}{\sigma_E\sqrt{1-\rho^2}F^E} & \frac{1}{\sigma_I\sqrt{1-\rho^2}F^I} \\
   \end{array}
  \right)\,.
 $$
 The quanto option pricing formula, in this setting, is given by
 \begin{equation}
 C_t=\mathbb E^{\mathbb Q}[g(F^E)h(\rho F^E +\sqrt{1-\rho^2}F^I)].
 \end{equation}
 By using the general formulae developed in the previous sections, we are able to compute analytically the values of different Malliavin weights. Here we set $a(t)=\frac{1}{T}$. We have
\begin{eqnarray*}
\pi^{\Delta_E}&=&\frac{1}{F^E(0)T}\int_0^T\frac{1}{\sigma_E}dW^E(t)-\frac{1}{F^E(0)T}\int_0^T\frac{\rho}{\sigma_E\sqrt{1-\rho^2}}d\widetilde{W}^I(t).
\end{eqnarray*}

\begin{eqnarray*}
\pi^{\Delta_I}&=&\frac{1}{F^I(0)T}\int_0^T\frac{1}{\sigma_I\sqrt{1-\rho^2}}d\widetilde{W}^I(t).
\end{eqnarray*}
\begin{eqnarray*}
\pi^{\Delta_{EI}}&=&\frac{1}{F^E(0)F^I(0)T^2}\left(\int_0^T\frac{1}{\sigma_E}dW^E(t)-\int_0^T\frac{\rho}{\sigma_E\sqrt{1-\rho^2}}d\widetilde{W}^I(t)\right)\left( \int_0^T\frac{1}{\sigma_I\sqrt{1-\rho^2}}d\widetilde{W}^I(t)\right)\\
&-&
\frac{1}{F^E(0)F^I(0)T^2}\int_0^T\frac{\rho}{\sigma_E\sigma_I(1-\rho^2)}dt.
\end{eqnarray*}
\section{Concluding remarks}
In this paper, we have derived the delta, the cross-gamma expectation formulae of the quanto energy option
written on a forward contract under the HJM framework. We have considered the
independent and the correlation cases to facilitate the residual risk analysis. The results are an extension of the work in Benth \emph{et al.} \cite{BenthLange} as they accommodate discontinuous payoff functions. It will be interesting to consider the case with stochastic volatility with a positive L\'evy processes when the market is incomplete. In Benth \emph{et al.} \cite{BenthGroth}, the authors analysed such a volatility model for a different payoff structure to the one considered in this paper. This will be considered in future research.
\subsection*{Acknowledgment}
The authors would like to thank Prof Fred E. Benth for discussing the modelling issues in
this paper. The work of R. K. was supported in part by the National Research Foundation of South Africa (Project No. 90313). The work of F. J. M. was supported in part by the National Research Foundation of South Africa (Grant Number: 105924).

 \end{document}